\documentclass[11pt]{article} 

\usepackage{fullpage}
\usepackage{times}
\usepackage{epsfig}
\usepackage{amsmath}
\usepackage{amssymb}
\usepackage{amstext}
\usepackage{xspace}
\usepackage{latexsym}
\usepackage{verbatim}
\usepackage{multirow}
\usepackage{float}

\newcommand{\argmax}{\operatorname{argmax}}

\newcommand{\expect}{\mathop{\operatorname{E}}}
\newcommand{\abs}[1]{\lvert{#1}\rvert}

\newcommand{\iid}{{i.~i.~d.\ }}

\newcommand{\val}{v}
\newcommand{\vals}{{\mathbf \val}}
\newcommand{\valsmi}{{\mathbf \val}_{-i}}
\newcommand{\valsmij}{{\vals}_{-ij}}
\newcommand{\vali}[1][i]{{\val}_{#1}}
\newcommand{\valij}[1][ij]{{\val}_{#1}}

\newcommand{\bid}{b}
\newcommand{\bids}{{\mathbf \bid}}

\newcommand{\dist}{F}

\newcommand{\disti}[1][i]{{\dist_{#1}}}

\newcommand{\dens}{f}

\newcommand{\densi}[1][i]{{\dens_{#1}}}

\newcommand{\price}{p}
\newcommand{\prices}{{\mathbf \price}}
\newcommand{\pricei}[1][i]{{\price_{#1}}}

\newcommand{\prob}{q}
\newcommand{\probs}{{\mathbf \prob}}
\newcommand{\probi}[1][i]{{\prob_{#1}}}

\newcommand{\ts}{{\mathbf t}}

\newcommand{\Vick}{\mathcal{V}}
\newcommand{\Mye}{\mathcal M}
\newcommand{\Opt}[1][]{\mathcal{M}_{#1}} 

\newcommand{\Lot}{\mathcal L}
\newcommand{\lotCol}{\mathcal L}
\newcommand{\lotI}[1][i]{{\lotCol}_{#1}}
\newcommand{\lotIJ}[1][ij]{{\lotCol}_{#1}}
\newcommand{\lotto}{\ell}
\newcommand{\I}{\mathcal I}

\newcommand{\Ifc}{\I^{\text{copies}}}
\newcommand{\sets}{{\mathcal J}}

\newcommand{\fudge}{u_{ij}(\valsmij)} 
\newcommand{\assign}{a} 

\newcommand{\LMech}{\mathcal{A}^{\lotCol}}

\newcommand{\MechL}{M^{\lotCol}}

\newcommand{\Mech}{M}

\newcommand{\MechLP}{M^{\lotCol'}}
\newcommand{\LMechP}{\mathcal{A}^{\lotCol'}}

\newcommand{\Rev}{{\mathcal R}}
\newcommand{\RevMye}{{\Rev}^{\Mye}}
\newcommand{\RevAlg}[1][\Alg]{{\Rev}^{#1}}

\newtheorem{theorem}{Theorem}

\newtheorem{claim}{Claim}
\newtheorem{lemma}[theorem]{Lemma}

\newtheorem{proposition}[theorem]{Proposition}
\newtheorem{corollary}[theorem]{Corollary}

\newcommand{\qed}{\mbox{\ \ \ }\rule{6pt}{7pt} \bigskip}

\renewcommand{\comment}[1]{}
\newenvironment{proof}{\noindent{\em Proof:}}{\hfill\qed}

\newcommand{\expand}[1]{\bar{#1}}

\begin{document}
\title{The power of randomness \\ in Bayesian optimal mechanism design}
\author{Shuchi Chawla\thanks{Computer Sciences Dept., University of Wisconsin -
  Madison. \tt{shuchi@cs.wisc.edu}.}
 \and David Malec\thanks{Computer Sciences Dept., University of Wisconsin -
  Madison. \tt{dmalec@cs.wisc.edu}.}
 \and Balasubramanian Sivan\thanks{Computer Sciences Dept., University
 of Wisconsin - Madison. \tt{balu2901@cs.wisc.edu}.}
}

\date{}

\maketitle
\begin{abstract}
We investigate the power of randomness in the context of a fundamental
Bayesian optimal mechanism design problem---a single seller aims to
maximize expected revenue by allocating multiple kinds of resources to
``unit-demand'' agents with preferences drawn from a known
distribution. When the agents' preferences are single-dimensional
Myerson's seminal work \cite{mye-81} shows that randomness offers no
benefit---the optimal mechanism is always deterministic. In the
multi-dimensional case, where each agent's preferences are given by
different values for each of the available services, Briest et
al.~\cite{BCKW-10} recently showed that the gap between the expected
revenue obtained by an optimal randomized mechanism and an optimal
deterministic mechanism can be unbounded even when a single agent is
offered only $4$ services. However, this large gap is attained through
unnatural instances where values of the agent for different services
are correlated in a specific way. We show that when the agent's values
involve no correlation or a specific kind of positive correlation, the
benefit of randomness is only a small constant factor ($4$ and $8$
respectively). Our model of positively correlated values (that we call
additive values) is a natural model for unit-demand agents and items
that are substitutes. Our results extend to multiple agent settings as
well.

\end{abstract}



\section{Introduction}

\begin{figure*}[t]
  \hfill
  \begin{minipage}[c]{.45\textwidth}
    \begin{center}
      \epsfig{file=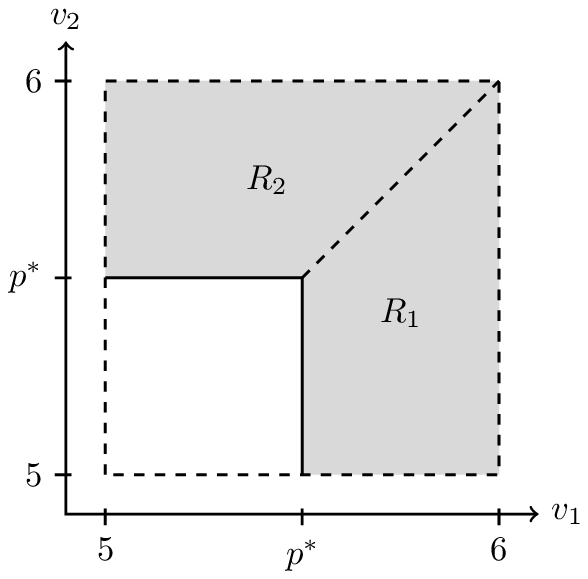}
    \end{center}
  \end{minipage}
  \hfill
  \begin{minipage}[c]{.45\textwidth}
    \begin{center}
      \epsfig{file=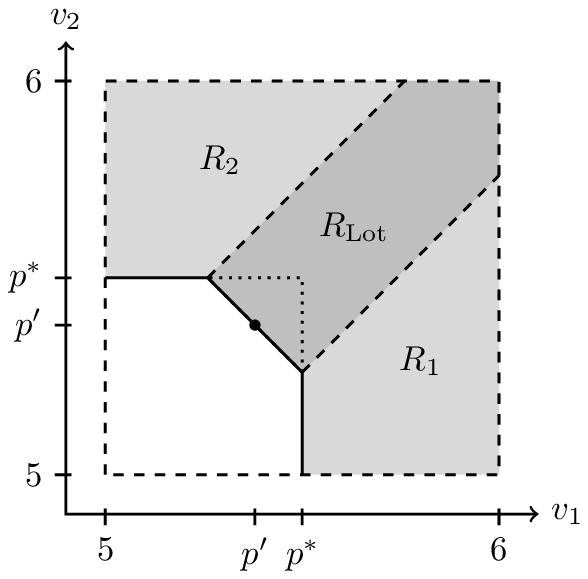}
    \end{center}
  \end{minipage}
\hfill
\caption{An example from \cite{Th-04} contrasting the optimal item pricing and the optimal lottery pricing. The regions $R_1$, $R_2$, and $R_{\text{Lot}}$ denote the sets of valuations at which the agent buys item $1$, item $2$, and the $(1/2,1/2)$ lottery respectively.\label{fig:pricing_vs_lottery}
}
\end{figure*}

A fundamental objective in the design of mechanisms is to
maximize the seller's revenue. In the absence of any information about
buyers' preferences, i.e. in prior-free settings, randomization is a
frequently used algorithmic technique (see, e.g., \cite{AGT-HK} and
references therein); In a spirit similar to randomness in online
algorithm design, it allows the seller to hedge against adversarial
values. While randomization unsurprisingly turns out to be essential
for any guarantees on revenue in certain prior-free settings, it
appears to be not so in Bayesian settings where the designer has
distributional information about the agents' types and the goal is to
maximize revenue in expectation over the distribution. For example,
for a single item auction in the Bayesian setting, Myerson's seminal
work~\cite{mye-81} shows that the optimal mechanism is always a
deterministic one.

In this work we investigate the power of randomness in the context of
the following archetypical multi-parameter optimal mechanism design
problem --- a single seller offers multiple kinds of service, and a
number of ``unit-demand'' agents are each interested in buying any one
of the services. Whereas in Myerson's work each agent has a
single-dimensional type (namely a value for the item under sale), in
our setting each agent has a multi-dimensional type characterized by a
(different) value for each of the services offered by the seller. An
example of such a setting is an online travel agency selling airline
tickets, hotel rooms, etc.  Customers have different preferences over
different available services, but are only interested in buying
one. We study the Bayesian version of this problem: the distribution
from which the buyers' preferences are drawn is known to the
seller. Given Myerson's observation about single-dimensional settings,
one might expect that in the multi-dimensional case the optimal
mechanism (ignoring computational issues) is once again deterministic.
Thanassoulis~\cite{Th-04} and Manelli and Vincent~\cite{MV-06}
independently discovered that this is not the case. This raises the
following natural question: {\em what quantitative benefit do
randomized mechanisms offer over deterministic ones in Bayesian
optimal mechanism design?}

To answer this question we must first understand the structure of
randomized mechanisms in multi-dimensional settings. In the context of
a single unit-demand agent and a seller offering multiple items, any
deterministic mechanism is simply a pricing for each of the items with
the agent picking the one that maximizes her utility (her value for
the item minus its price). Likewise, randomized mechanisms can be
thought of as pricings for distributions or convex combinations over
items. These convex combinations are called {\em lotteries}. A
risk-neutral buyer with a quasiconcave utility function buys the
lottery that maximizes his expected value minus the price of the
lottery.

The following example due to Thanassoulis explains how lotteries
work. Suppose that a seller offers two items for sale to a single
buyer, and that the buyer's value for each of the items is
independently uniformly distributed in the interval $[5,6]$. The
optimal deterministic mechanism for the seller is to simply price each
of the items at $p^*=\$ 5.097$ (see
Figure~\ref{fig:pricing_vs_lottery}).  In a randomized mechanism, the
seller may in addition price a $(1/2,1/2)$ distribution over the two
items at a slightly lower price of $p'=\$5.057$. If the buyer buys
this lottery, the seller tosses a coin and allocates the first item to
her with probability $1/2$ and the second with probability $1/2$. A
buyer that is nearly indifferent between the two items would prefer to
buy the lottery because of its lower cost, than either one of the
items. While the seller loses some revenue by selling the lower priced
lottery with some probability, he gains by selling to a larger segment
of the market (those that cannot afford either of the individual items
but can afford the lower priced lottery). In this example the gain is
more than the loss, so that introducing the lottery improves the
seller's revenue. As this example indicates, lotteries help in optimal
mechanism design by giving the seller more latitude to price
discriminate among buyers with different preferences.

In general, a randomized mechanism can offer to the buyer a menu of
prices for arbitrarily many lotteries. We call such a menu a {\em
  lottery pricing}, and likewise a deterministic pricing an {\em item
  pricing}. While in multiple agent settings randomized mechanisms can
be more complicated, we show that any such mechanism can be
interpreted as offering to each agent simultaneously a lottery pricing
that is a function of values of other agents.

The question of whether and to what extent randomization helps in
Bayesian optimal mechanism design is not merely a pedantic
one. Mechanisms similar to lottery pricings are seen in practice. For
example, the website priceline.com routinely sells airline tickets to
customers without disclosing at the time of sale crucial details such
as the time of travel, carrier, etc. While customers are unaware of
the distribution from which the final service is picked, the tradeoffs
for customers are similar---the uncertainty in the quality of the
final item against the cheaper price. Travel agencies offering
vacation packages use similar devices.

Until recently, the largest gap known between item pricings and
lottery pricings for a single agent was a gap of $3/2$ due to
Pavlov~\cite{Pav-06}; For the special case where values for different
items are independent, Thanassoulis gave the best gap example with a
gap of $1.1$.
Recently Briest et al.~\cite{BCKW-10} showed that in single-agent
settings in fact the gap between lottery pricings and item pricings
can be unbounded even with only $4$ items. However the value
distributions for which such gaps are achieved are quite unnatural
with the values of different items being highly correlated. In this
paper we show that the gap between lottery pricings and item pricings
is small for distributions involving limited correlation between
items.

We further extend these results to the multiple-agent setting with the
seller facing a general feasibility constraint, obtaining the first
results of this kind. Mechanism design in the multiple-agent
multi-parameter setting is poorly understood~\cite{MV-07}. Until recently
there were no general characterizations for optimal or approximately
optimal mechanisms similar to Myerson's for the single-parameter
case. Chawla et al.~\cite{CHMS-10} recently developed constant-factor
approximations to
optimal {\em deterministic} mechanisms in this setting for a certain
class of feasibility constraints (namely matroids and related set
systems). We extend their results to show that their (deterministic)
mechanisms achieve a constant factor approximation with respect to the
optimal randomized mechanism as well, again implying a small gap
between randomized and deterministic mechanisms.

\begin{figure*}[t]
  \hfill
  \begin{minipage}[c]{.45\textwidth}
    \begin{center}
      \epsfig{file=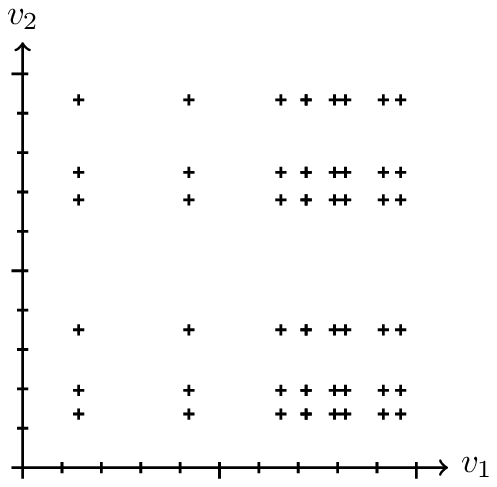}
    \end{center}
  \end{minipage}
  \hfill
  \begin{minipage}[c]{.45\textwidth}
    \begin{center}
      \epsfig{file=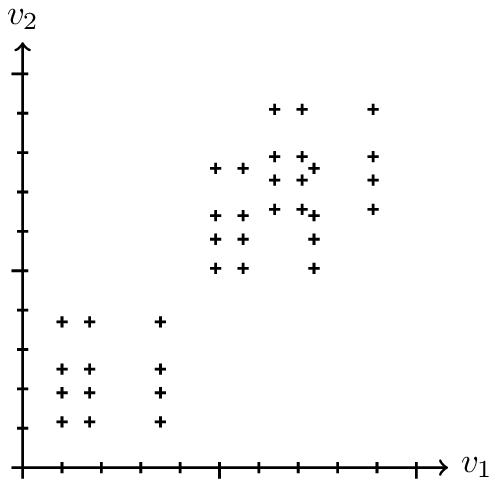}
    \end{center}
  \end{minipage}
\hfill
\caption{An example of a product distribution for valuations
  contrasted against an example of an additive distribution.\label{fig:unif_vs_corr}
}
\end{figure*}

\subsubsection*{Our results and techniques}
We follow a technique introduced in \cite{CHK-07} for relating
multi-parameter mechanisms to mechanisms for a related
single-parameter problem. Chawla, Hartline and Kleinberg~\cite{CHK-07}
relate a single unit-demand agent $m$-item mechanism design problem to
an $m$-agent single-item auction setting, by ``splitting'' the
unit-demand agent into $m$ independent ``copies''. They argue that the
increased competition among copies benefits the seller and leads to
higher revenue. Formally, given an item pricing $\prices$ they
construct a truthful mechanism $A^{\prices}$ that allocates the item
to agent $i$ whenever $\prices$ allocates item $i$ to the
multi-parameter agent (that is, $A^{\prices}$ has the ``same''
allocation rule as $\prices$). They then argue that the price that
$A^{\prices}$ charges is no less than the price that $\prices$ charges
for any instantiation of values. Therefore, {\em the expected revenue
of the optimal multi-parameter mechanism is bounded above by the
expected revenue of Myerson's mechanism for the related
single-parameter problem with copies}. Chawla et al. use this upper
bound to design an item pricing for the multi-parameter problem with
revenue within a factor of $3$ of the expected revenue of Myerson's
mechanism for the instance with copies, thereby obtaining a
$3$-approximation to the optimal deterministic mechanism for the
single-agent problem.

Unfortunately the upper bound of the expected revenue of Myerson's
mechanism does not hold for randomized mechanisms. The appendix gives
an example where the revenue of Myerson's mechanism for the instance
with copies is a factor of $1.13$ smaller than that of the optimal
lottery pricing for the multi-parameter problem. In fact, the
mechanism $A^{\Lot}$ with the ``same'' allocation rule as a lottery
pricing $\Lot$ may obtain zero revenue even when the lottery pricing
obtains non-zero revenue. Our main result is that this gap between
Myerson's mechanism and the optimal lottery pricing is no larger than
a factor of $2$. Specifically, given a lottery pricing, we can
construct two mechanisms, one being $A^{\Lot}$ and the other a Vickrey
auction, such that the sum of the revenues of the two mechanisms is an
upper bound on the revenue of the lottery pricing. Combining this with
the result of Chawla, Hartline and Kleinberg (and an improvement over
it in \cite{CHMS-10}), we get that for a single unit-demand agent
multi-parameter problem, the gap between lottery pricings and item
pricings is at most $4$.

Chawla et al.'s result as well as our factor-of-$4$ gap holds for
instances where the values of the agent for different items are
independent. For a unit-demand agent, this independence assumption is
unrealistic. However, on the other end of the spectrum, Briest et
al. show that with arbitrary correlations between item values, the gap
can be unbounded. We therefore examine the following natural model for
values involving limited correlation. The type of the unit-demand
agent is $m+1$ dimensional --- $(t_0, t_1, \cdots, t_m)$; the agent's
value for item $i$ is $\vali=t_0+t_i$. Here $t_0$ can be thought of as
the buyer's ``base'' value for obtaining any of the items, and the
$t_i$'s represent the buyer's perceived quality of the different
items. This additive value distribution introduces a positive
correlation between values of different items\footnote{This model is
similar to ``multiplicative'' value distributions that have been
studied previously in the context of bundle pricing problems (see,
e.g., \cite{Arm99}).}. Figure~\ref{fig:unif_vs_corr} shows an example
of one such discrete distribution contrasted against a product
distribution.


In this additive distribution setting we show that the gap between
randomized and deterministic mechanisms is at most a factor of
$8$. Once again our approach is to start with an optimal lottery
pricing for the multi-parameter instance, construct an ensemble of
mechanisms based on it for the related single-parameter instance, and
then construct a pricing for the multi-parameter instance based on the
mechanisms.

Our results extend to multi-agent settings as well. The simplest
multi-agent setting we consider involves $n$ agents and $m$ items
(with copies), where the seller faces a supply constraint for each of
the items. A feasible allocation is a matching between agents and
items that respects multiplicities of items. More generally, we
consider settings where the seller faces a matroid feasibility
constraint---any feasible allocation must be an independent set in a
given matroid in addition to allocating at most one item per agent
(see Section~\ref{subsec:multiItemAuctions} for the definition of a
matroid). In both these cases we show that the gap between the
expected revenue of the optimal randomized and the optimal
deterministic mechanisms is a small constant factor. Once again we
rely on the approach of relating the multi-parameter instance to a
single-parameter instance where each unit-demand agent is split into
multiple selfish ``pseudo-agents''. This approach was first developed
in \cite{CHMS-10}. In particular we showed in \cite{CHMS-10} that for
the settings described above, there exist deterministic mechanisms
that obtain revenue within a constant factor of the revenue of
Myerson's mechanism for the related single-parameter instance. In
Section~\ref{sec:multi} we show that the revenue of any randomized
mechanism for these settings can be bounded from above by $5$ times
the revenue of Myerson's mechanism for the single-parameter instance.
The challenge in these settings is to ensure that the mechanisms that
we construct satisfy the non-trivial feasibility constraint that the
seller faces.

\subsubsection*{Related work}
As mentioned earlier, randomness is used extensively in prior-free
mechanism design (see, e.g., \cite{AGT-HK} and references
therein). While symmetric deterministic mechanisms provably cannot
obtain any guarantees on revenue in that setting, Aggarwal et
al.~\cite{AFGHIS-05} show that by exploiting asymmetry prior-free
mechanisms can be derandomized at a constant factor loss in revenue.

Our mechanism design setting with unit-demand agents is closely
related to the standard setting for envy-free pricing problems
considered in literature~\cite{GHKKKM-05, bb-06, BBM08, Bri-06,
CHK-07}; Those works study the single-agent problem with a correlated
value distribution and aim to approximate the optimal deterministic
mechanism (item pricing). Our single-agent setting is most closely
related to the work of Chawla, Hartline and Kleinberg~\cite{CHK-07}
who gave a $3$ approximation to the optimal deterministic mechanism
for single-agent product-distribution instances, and builds upon
techniques developed in that work.

In economics literature, the study of Bayesian optimal mechanisms has
focused on deterministic mechanisms. It is well-known that for
single-parameter instances the optimal mechanism is
deterministic \cite{mye-81, RZ-83}. Following Myerson's
result \cite{mye-81} for single-parameter mechanisms, there were a
number of attempts to obtain simple characterizations of optimal
mechanisms in the multi-parameter setting \cite{MM-88, RC98, MV-07},
however no general-purpose characterization of such mechanisms is
known~\cite{MV-07}. Recently Chawla et al.~\cite{CHMS-10} gave the
first approximations to optimal deterministic mechanisms for a large
class of multi-parameter problems. This paper extends techniques
developed in that work and one of the implications of our work is that
the mechanisms developed in \cite{CHMS-10} are approximately-optimal
with respect to the optimal randomized mechanisms as well.

The study of the benefit of randomness in multi-parameter mechanism
design was initiated by Thanassoulis~\cite{Th-04} who presented
single-parameter instances with valuations drawn from product
distributions where randomness helps increase the revenue by about
$8$-$10\%$. Manelli and Vincent~\cite{MV-06} and Pavlov~\cite{Pav-06}
presented other examples with small gaps. Briest et al.~\cite{BCKW-10}
were the first to uncover the extent of the benefit of
randomization. They showed that lottery pricings can be arbitrarily
better than item pricings in terms of revenue even for the case of $4$
items offered to a single agent.

\section{Definitions and problem set-up}

\newlength{\indentFix}
\setlength{\indentFix}{.8cm}





\subsection{Bayesian optimal mechanism design}
\label{subsec:bayesOPT}
We study the following mechanism design problem. There is one seller
and $n$ buyers (agents) indexed by the set $I$. The seller offers $m$
different services indexed by the set $J$. Agents are risk-neutral and
are each interested in buying any one of the $m$ services. Agent $i$
has value $\vali[ij]$ for service $j$ which is a random variable. We
use $\valsmi$ to denote the vector of values of all agents except
agent $i$. The seller faces no costs for providing service, but must
satisfy certain feasibility constraints (e.g. supply constraints in a
limited supply setting). We represent these feasibility constraints as
a set system $\sets$ over pairs $(i,j)$, that is, $\sets\subseteq
2^{I\times J}$. Each subset of $I\times J$ in $\sets$ is a feasible
allocation of services to agents. 

The seller's goal is to maximize her revenue in expectation over the
buyers' valuations. We call this problem the {\em Bayesian
  multi-parameter unit-demand (optimal) mechanism design} problem
(BMUMD). A deterministic mechanism for this problem maps any set of
bids $\bids$ to an allocation $M(\bids)\in \sets$ and a pricing
$\pi(\bids)$ with a price $\pi_i$ to be paid by agent $i$. A
randomized mechanism maps a set of bids to a distribution over
$\sets$; we use $M(\bids)$ to denote this distribution over $I\times
J$. 

We focus on the class of incentive compatible mechanisms and will
hereafter assume that $\bids=\vals$. We use $\RevAlg[\Mech](\vals)$ to
denote the revenue of a mechanism $\Mech$ at valuation vector $\vals$:
$\RevAlg[\Mech](\vals) = \sum_{i\in I} \pi_i(\vals)$ where $\pi$ is
the pricing rule for $\Mech$.  To aid disambiguation, we sometimes use
$\RevAlg[\Mech]_i(\vals)$ to denote $\pi_i(\vals)$ for $\Mech$.  The
expected revenue of a mechanism is $\RevAlg[\Mech] =
\expect_{\vals}[\RevAlg[\Mech](\vals)]$.

\addtolength{\leftmargini}{\indentFix}

We consider the following special cases of the BMUMD:
\begin{itemize}
\item[\bf Setting 1: ] {\bf Single agent with independent values}. The
  agent values item $j$ at $\vali[j]$, which is an independent random
  variable with distribution $\disti[j]$ and density $\densi[j]$.
\item[\bf Setting 2: ] {\bf Single agent with additive values}. There are $m$ items,
  and the agent's type, $\{t_0,\cdots,t_m\}$, is $m+1$
  dimensional. $t_j$ is distributed independently according to
  $\disti[j]$. The agent's value for item $j$ is $\vali[j]=t_0+t_j$.
\item[\bf Setting 3: ] {\bf Multiple agents and multiple items with independent
  values}. There are $n$ agents and $m$ items. Agent $i$'s value for
  item $j$, $\vali[ij]$, is distributed independently according to
  $\disti[ij]$. Any matching between items and agents is feasible.
\item[\bf Setting 4: ] {\bf Multiple unit-demand agents with matroid
  feasibility constraint}. There are $n$ agents and $m$
  services. Agent $i$'s value for item $j$, $\vali[ij]$, is
  distributed independently according to $\disti[ij]$. The set system
  $\sets$ is an intersection of a matroid with the unit-demand
  constraints for the agents and is thus a generalization of the
  previous matching setting. (See Section~\ref{subsec:multiItemAuctions} 
for the definition of
  a matroid.)
\end{itemize}
\addtolength{\leftmargini}{-\indentFix}


\subsubsection*{Single-parameter mechanism design}
The single-parameter version of the Bayesian optimal mechanism design
problem (abbreviated BSMD) is stated as follows. There are $n$
single-parameter agents and a single seller providing a certain
service. Agent $i$'s value $\vali$ for getting served is a random
variable. We use $\valsmi$ to denote the vector of values of all
agents except agent $i$. The seller faces a feasibility constraint
specified by a set system $\sets\subseteq 2^{[n]}$, and is allowed to
serve any set of agents in $\sets$. As in the multi parameter case, a
mechanism $M$ for this problem is a function that maps a vector of
values $\vals$ to an {\em allocation} $M(\vals)\in \sets$ and a {\em
  pricing} $\pi(\vals)$. Myerson's seminal work describes the revenue
maximizing mechanism for BSMD; this optimal mechanism is
deterministic.

\subsection{Relating multi-parameter MD to single-parameter MD}
\label{subsec:md-redn}

In previous work~\cite{CHMS-10} we presented a general reduction from
the multi-parameter optimal mechanism design problem to the
single-parameter setting. This approach begins with defining an
instance $\Ifc$ of the BSMD given an instance $\I$ of the BMUMD. Our
previous work then shows that for several kinds of feasibility
constraints there exists a deterministic mechanism for $\I$ with
revenue at least a constant fraction of that of the optimal mechanism
for $\Ifc$. We state these results below without proof.

We begin by describing the instance $\Ifc$. Let $\I$ be an instance of
the BMUMD with $n$ agents and a single seller providing $m$ different
services, and with feasibility constraint $\sets$. We define a new
instance of the BSMD in the following manner. We split each agent in
$\I$ into $m$ distinct agents (hereafter called ``copies'' or
``pseudo-agents''). Each pseudo-agent is interested in a single item
$j\in [m]$ and behaves independently of (and potentially to the
detriment of) other pseudo-agents. Formally, the instance has $mn$
distinct pseudo-agents each interested in a single service;
pseudo-agent $(i,j)$'s value for getting served, $\vali[ij]$, is
distributed according to $\disti[ij]$. The mechanism again faces a
feasibility constraint given by the set system $\sets$.

$\Ifc$ is similar to $\I$ except that it involves more competition
(among different pseudo-agents corresponding to the same
multi-parameter agent). Therefore it is natural to expect that a
seller can obtain more revenue in the instance $\Ifc$ than in
$\I$. The following results show that in Settings 1 and 3 it cannot
obtain too much more.


\begin{theorem}(Theorem 4 and 10 in \cite{CHMS-10})
\label{thm:bmumdPricing}
Given an instance $\I$ of the single agent BMUMD (Setting 1), there
exists a truthful deterministic mechanism for $\I$, whose revenue is
at least 1/2 of the revenue of any truthful mechanism for the instance
$\Ifc$.
\end{theorem}

\begin{theorem}(Theorem 14 in \cite{CHMS-10})
\label{thm:bmumdOPM}
Given an instance $\I$ of the BMUMD with multiple agents and multiple
items (Setting 3), there exists a truthful deterministic mechanism for
$\I$, whose revenue is at least $4/27$th of the revenue of any
truthful mechanism for the instance $\Ifc$.
\end{theorem}

In Setting 4, \cite{CHMS-10} obtain a somewhat weaker result comparing
the revenue of an incentive-compatible mechanism for $\Ifc$ to that of
a deterministic mechanism for $\I$ that is not truthful but is an {\em
  implementation in undominated strategies} \cite{BLP09}. Formally,
for an agent $i$, a strategy $s_i$ is said to be dominated by a
strategy $s'_i$ if for all strategies $s_{-i}$ of other agents, the
utility that $i$ obtains from using $s_i$ is no better than that from
using $s'_i$, and for some strategy $s_{-i}$, it is strictly worse. A
mechanism is an algorithmic implementation of an
$\alpha$-approximation in undominated strategies if for every outcome
of the mechanism where every agent plays an undominated strategy, the
objective function value of the mechanism is within a factor of
$\alpha$ of the optimal, and every agent can easily compute for any
dominated strategy a strategy that dominates it.

\begin{theorem}(Theorem 17 in \cite{CHMS-10})
\label{thm:bmumdSPM}
Given an instance $\I$ of the BMUMD with unit-demand agents and a
general matroid constraint (Setting 4), there exists a deterministic
mechanism for $\I$ implemented in undominated strategies, whose
revenue is at least $1/8$th of the revenue of any truthful mechanism
for the instance $\Ifc$.
\end{theorem}

\section{Lotteries and randomized mechanisms}

We now define a class of mechanisms for the BMUMD that will be useful
in our analysis. The following subsection shows that this class
encompasses arbitrary randomized mechanisms.

\subsection{Lotteries or random allocations}

An $m$-dimensional {\em lottery} is a vector
$\lotto=(q_1,\cdots,q_m,p)$ where $p$ is the price of the lottery and
$(q_1,\cdots,q_m)$ is a probability distribution over $m$ items,
$\sum_{j\in [m]} q_j\le 1$. A {\em lottery pricing}
$\Lot=\{\lotto_1,\lotto_2,\cdots\}$ is a randomized selling mechanism
for $m$ items targeted towards a single unit-demand buyer where the
buyer is offered a collection of (an arbitrary number of)
lotteries. The buyer can select any one or no lottery from the
collection, and is then allocated an item drawn from the probability
distribution defined by the lottery and charged the price of the
lottery. A rational risk-neutral buyer selects the lottery that
maximizes her utility: $\sum_{j\in [m]} \probi[j]\vali[j] - p$.

A {\em lottery-based mechanism} $M^{\lotCol}$ for $m$ services
targeted towards $n$ agents is a randomized selling mechanism defined
through an ensemble of lottery pricings $\lotCol$. $M^{\lotCol}$ and
$\lotCol$ satisfy the following properties:
\begin{enumerate}
\item For every instantiation of values of the agents $\vals$,
$\lotCol$ contains $n$ lottery pricings,
$\Lot_1(\vals), \cdots, \Lot_n(\vals)$, where $\Lot_i(\vals)$ is an
$m$-dimensional lottery pricing targeted toward agent $i$.
\item $\Lot_i(\vals)$ is a function of $\valsmi$, the values of all
agents other than agent $i$.
\item The mechanism $M^{\lotCol}$ is implemented as follows. It first
elicits bids $\bids$ from agents, and then offers to agent $i$
(simultaneously with other agents) the lottery pricing
$\Lot_i(\bids)$. Let $\lotto_i(\bids)$ denote the lottery picked by agent
$i$ and let $\probi[ij](\bids)$ denote the probability with which
lottery $\lotto_i(\bids)$ offers service $j$ to agent $i$. Agent $i$ is
allocated item $j$ with probability $\probi[ij](\bids)$.\footnote{Note
that these allocations to agents are not necessarily done
independently; The feasibility constraint may require correlations
between items allocated to different agents. However these details do
not affect our analysis, so we ignore them.}
\item The probabilities $\probi[ij](\vals)$ satisfy the following
feasibility constraint:
\begin{align*}
\sum_{(i,j)\in S} \probi[ij](\vals) \leq r(S),\ \forall
S \subseteq I\times J,\ \forall \vals
\end{align*}
where $r(S)$ is the cardinality of some maximum sized feasible subset
of $S$.\footnote{This condition is weaker than may be necessary for
  certain kinds of set systems, but suffices for our purpose.}
\end{enumerate}

\subsection{Randomized mechanisms as lotteries}
\label{subsec:randMechAsLotteries}
We now show that every truthful randomized mechanism for the BMUMD can
be interpreted as a truthful lottery-based mechanism.

\begin{lemma}\label{lem:randMechLotRedn}
Every incentive-compatible randomized mechanism for a multi-agent
BMUMD problem is equivalent to a lottery-based mechanism.
\end{lemma}
\begin{proof}
Given a mechanism $M$ with randomized allocation rule $M(\vals)$ and
pricing rule $\pi(\vals)$ we define a lottery-based mechanism as
follows. Consider an agent $i$ and a fixed instantiation of
$\valsmi$. Then for every instantiation of $\vals_i$, consider the
probabilities with which $M$ allocates service $j$ to agent $i$, as
well as the prices that $M$ charges. Each such probability vector
along with the corresponding price forms a lottery in
$\Lot_i(\valsmi)$ in the new mechanism. Formally, $\Lot_i(\valsmi) =
\{(\probs_i,\pricei)\ |\ \exists \vals_i \text{ with } \probs_i =
M_i(\valsmi,\vals_i) \text{ and } \pricei =
\pi_i(\valsmi,\vals_i)\}$.

We now claim that the allocation rule and pricing rule of the new
mechanism is precisely the same as the old mechanism. Suppose
not. Then at some valuation vector $\vals$ and for some agent $i$,
$(\probs_i(\valsmi,\vals_i), \pricei(\valsmi,\vals_i)) \ne
(M_i(\valsmi,\vals_i), \pi_i(\valsmi,\vals_i))$, where the former is
the allocation and price rule for the lottery-based mechanism and the
latter the allocation and price rule for the original mechanism
$M$. But, given our construction, $(\probs_i(\valsmi,\vals_i),
\pricei(\valsmi,\vals_i)) = (M(\valsmi,\vals'_i),
\pi_i(\valsmi,\vals'_i))$ for some other value vector $\vals'_i$ for
agent $i$. But this implies that in $M$ agent $i$ can benefit from
lying and reporting $\vals'_i$ when the true value vector is
$\vals$. This contradicts the incentive compatibility of $M$.
\end{proof}

\subsection{A mechanism for $\Ifc$ based on lotteries}
\label{subsubsec:copyMech}
As noted earlier, our main technique is to relate the revenue of
lottery-based mechanisms for an instance $\I$ of the BMUMD to the
optimal mechanism for a related instance $\Ifc$ of the BSMD. We now
describe a mechanism for $\Ifc$ based on a given lottery-based
mechanism for $\I$.

Consider an instance $\I$ of the BMUMD. Given a lottery-based
mechanism $\MechL$ for $\I$ that uses the ensemble of lottery pricings
$\lotCol$, we define a mechanism $\LMech$ for the instance $\Ifc$.

Based on $\lotCol$, the mechanism $\LMech$ forms a one dimensional
lottery pricing for each of the $mn$ pseudo-agents. The lottery
pricing offered to pseudo-agent $(i,j)$, which we denote $\lotIJ$, is
a function of $\valsmij$ and is derived from the lottery pricing
$\lotI\in\lotCol$ as follows.  Given a valuation vector $\valsmij$,
for each $\lotto=(\probi[i1], \probi[i2],\dots,
\probi[im],\price)\in\lotI(\valsmi)$, 
$\LMech$ adds a lottery $\lotto_j=(\prob',\price')$ to $\lotIJ$
defined by
\begin{align*}
\prob' &= \probi[ij]\text{; and}\\
\price'&= \price - \sum_{k\neq j}\probi[ik]\vali[ik] +\fudge,
\end{align*}
where the term $\fudge\ge0$ is chosen to be the least
value ensuring that the lottery preferred by pseudo-agent $(i,j)$ when
$\valij=0$ (if any) has a non-negative price.

We note the following properties of $\LMech$:
\begin{enumerate}
\item (truthfulness) That $\LMech$ is truthful follows immediately
  from the fact that the one dimensional lottery pricing $\lotIJ$
  offered to pseudo-agent $(i,j)$ does not depend on $\vali[ij]$, and
  the pseudo-agent may choose any lottery from $\lotIJ$.
\item (allocation rule) Suppose first that for $(i,j)$ and some
  $\valsmij$, $\fudge=0$. Then for any $\vali[ij]$, the utility of
  pseudo-agent $(i,j)$ from lottery $\lotto_j\in\lotIJ$ is the same as
  utility of agent $i$ from lottery $\lotto\in \Lot_i$. Therefore with
  $\fudge=0$, in $\MechL$ agent $i$ purchases lottery
  $\lotto\in\Lot_i$ if and only if, in $\LMech$ the pseudo-agent
  $(i,j)$ purchases lottery $\lotto_j\in\lotIJ$. Moreover, since the
  price shifts $\fudge$ we apply are the same for every lottery
  offered to $(i,j)$, the only manner in which preferences can change
  is if the pseudo-agent obtains negative utility from his preferred
  lottery, in which case he chooses to buy no lottery at all. However,
  our choice of $\fudge$ ensures that the agent obtains non-negative
  utility at $\vali[ij]=0$ and thus also at arbitrary $\vali[ij]$, and
  so the allocation rule of $\LMech$ is identical to that of $\MechL$.
\item (feasibility) Feasibility follows immediately from the fact that
  $\MechL$ satisfies feasibility and the allocation
  rules of the two mechanisms are identical.
\item (nonnegative revenue) Our choice of $\fudge$ ensures that the
  revenue $\LMech$ receives from each agent is always nonnegative;
  this is critical in later arguments, since it allows us to claim
  that the revenue that $\LMech$ obtains from any subset of the
  pseudo-agents is bounded from above by the total expected revenue of
  $\LMech$.
\end{enumerate}

We now relate the revenues of $\MechL$ and $\LMech$.  Let $\assign$ be
any function carrying valuation vectors to sets of pseudo-agents which
respects the unit-demand constraint, i.e. for any valuation vector
$\vals$, for each $i\in{I}$ there exists at most one $j\in{J}$ such
that $(i,j)\in\assign(\vals)$. We call such a function a unit-demand
allocation function. Then we get the following lemma.

\begin{lemma}
\label{lemma:copiesBoundForRevML}
  For any valuation vector $\vals$ and any unit-demand allocation
  function $\assign(\vals)$, we have 
\begin{align*}
  \RevAlg[M^{\lotCol}](\vals)
  &\le \sum_{(i,j)\in\assign(\vals)}\RevAlg[\LMech]_{ij}(\vals)+\sum_{(i,j)\notin\assign(\vals)}\probi[ij](\vals)\vali[ij]\\
  &\le \RevAlg[\LMech](\vals) + \sum_{(i,j)\notin\assign(\vals)}\probi[ij](\vals)\vali[ij],
\end{align*}
  where
  $\lotto_i(\vals)=(\probi[i1](\vals),\dots,\probi[im](\vals),\price_i(\vals))$
  is the lottery purchased by agent $i$ at valuation $\vals$ in the
  mechanism $\MechL$.
\end{lemma}
\begin{proof}
 The revenue $\RevAlg[\MechL](\vals)$ of the lottery-based mechanism
 $\MechL$ at $\vals$ can be written as the sum of the revenues from
 the constituent lottery pricings:
\begin{align*}
\RevAlg[\MechL](\vals) = \sum_{i=1}^{n}\RevAlg[\MechL]_i(\vals).
\end{align*}
If we define $\lotto_i(\vals) =
(\probi[i1](\vals),\dots,\probi[im](\vals),\price_i(\vals)) \in \lotI$
to be the lottery chosen by agent $i$ at $\vals$, then
$\RevAlg[\MechL]_i(\vals)$, which is just the price $\price_i(\vals)$,
can be written as
\begin{equation}
  \label{eqn:lotRevSplit}
  \begin{aligned}
    \RevAlg[\MechL]_i(\vals) 
    &= \left(\price_i(\vals) - \sum_{k\neq j}\probi[ik](\vals)\vali[ik]\right) 
    + \sum_{k\neq j}\probi[ik](\vals)\vali[ik] \\
    &\le \RevAlg[\LMech]_{ij}(\vals) + \sum_{k\neq j}\probi[ik](\vals)\vali[ik],
  \end{aligned}
\end{equation}
for any $j$, where $\RevAlg[\LMech]_{ij}(\vals)$ is the revenue of
mechanism $\LMech$ from the pseudo-agent $(i,j)$.  Furthermore, since
agent $i$ would never elect to purchase a lottery yielding negative
utility, we also have that
\begin{equation}
  \label{eqn:lotRevVsSocVal}
  \RevAlg[\MechL]_i(\vals) \le \sum_{k}\probi[ik](\vals)\vali[ik].
\end{equation}
Note that we designed $\LMech$ such it receives nonnegative revenue
from every pseudo-agent, and $\assign$ contains at most one
pseudo-agent $(i,j)$ for any $i$; so by applying one
of~\eqref{eqn:lotRevSplit} or~\eqref{eqn:lotRevVsSocVal} for each $i$
according to which pseudo-agents $\assign(\vals)$ contains, we get
that
\begin{align*}
  \RevAlg[M^{\lotCol}](\vals)
  &\le \sum_{(i,j)\in\assign(\vals)}\RevAlg[\LMech]_{ij}(\vals)+\sum_{(i,j)\notin\assign(\vals)}\probi[ij](\vals)\vali[ij]\\
  &\le \RevAlg[\LMech](\vals) + \sum_{(i,j)\notin\assign(\vals)}\probi[ij](\vals)\vali[ij],
\end{align*}
the claimed bound.
\end{proof}

\section{Single-agent setting}
In this section we focus on instances of the BMUMD involving a single
agent and $m$ items. In the single agent setting, randomized and
deterministic mechanisms become simply lotteries and pricings,
respectively.  Briest et al. \cite{BCKW-10} demonstrated that when values for
different items are arbitrarily correlated, it is possible to
construct examples where the ratio between the optimal expected
revenues from lotteries and pricings is unbounded. We show that in the
absence of such correlation this ratio is small. Specifically, when
values are distributed independently, the ratio is no more than $4$
(Section~\ref{sec:single-indep}). Moreover, when values have a certain
kind of positive correlation (additive values; Setting 2 described in
Section~\ref{subsec:bayesOPT}), the ratio is at most $8$
(Section~\ref{sec:single-additive}).

\subsection{Independent values (Setting 1)}
\label{sec:single-indep}

We first analyze Setting 1, that is where the value of the agent for
item $i$, $\vali$, is independently distributed according to
c.d.f. $\disti$. Given an instance $\I$ of the single agent BMUMD,
consider the form of the associated instance $\Ifc$.  Note that while
each pseudo-agent desires a different item, the fact that only one
item may be sold means they are effectively competing for the same
thing, the privilege of being served.  Thus, $\Ifc$ can be thought of
as being in a single-item auction setting.  This observation leads to
the following lemma.

\begin{lemma}
  \label{lemma:lottoVsPrice} For any instance $\I$ of the BMUMD in
  Setting 1, the revenue of the optimal deterministic mechanism is at
  least one-fourth the revenue of the optimal randomized mechanism.
\end{lemma}
\begin{proof}
As previously observed, any randomized mechanism in the single-agent
setting is precisely a lottery pricing $\Lot$.  Let the mechanism
$\LMech$ be as described in Section~\ref{subsubsec:copyMech}.
Applying Lemma~\ref{lemma:copiesBoundForRevML} with
$\assign(\vals)=i^\ast=\argmax_{i}(\vali)$ yields
\begin{align*}                      
  \RevAlg[\Lot](\vals)                            
  &\le \RevAlg[\LMech](\vals) + \sum_{i\neq i^\ast}\probi(\vals)\vali \\
  &\le \RevAlg[\LMech](\vals) + \max_{i\neq i^\ast}\vali,
\end{align*}
since the $\probi(\vals)$'s sum to at most one.  The key observation
is that the second term is precisely the revenue that the Vickrey
auction $\Vick$ would achieve in the instance $\Ifc$ given bids
$\vals$; so we get that in expectation
\begin{equation*}
  \RevAlg[\Lot] \le \RevAlg[\LMech] + \RevAlg[\Vick],
\end{equation*}
and need only apply Theorem~\ref{thm:bmumdPricing} to prove the lemma.
\end{proof}

\subsection{Additive values (Setting 2)}
\label{sec:single-additive}
We demonstrate that a result similar to that of the previous section
holds even in the presence of certain types of correlation.  Consider
again the single agent setting; since the agent is unit demand, it
makes sense to think of the services being offered as perfect
substitutes.  A natural form of correlation, then, would be for the
agent to have some ``base'' value for being served (regardless of
which service is received), plus an additive value specific to the
particular service received.

The setting we consider modifies the single-agent setting by making
agent types consist of $(m+1)$ independently distributed values
$\{t_0,t_1,\dots,t_m\}$; now, the agent's value for item $i$ becomes
$\vali=t_i+t_0$.

Let $\Lot$ be a lottery system over $m$ items in the additive
setting described.  We have the following lemma.
\begin{lemma}
  Given an instance $\I$ of the BMUMD in Setting 2, the revenue
  of any lottery system $\Lot$ for $\I$ satisfies
  $\RevAlg[\Lot]\le 8\RevAlg[\prices]$, for some pricing $\prices$ for
  $\I$.
\end{lemma}
\begin{proof}
  We begin by demonstrating a bound with a weaker multiplicative
  factor of $9$ and then show how to improve it to a factor of $8$.
  Our main technique is to consider an uncorrelated setting $\I'$
  derived from $\I$.  We define $\I'$ to be a single agent setting
  with $(m+1)$ items, and interpret the values $\{t_0,\dots,t_m\}$
  making up an agent's type in $\I$ as being the values of the agent
  in setting $\I'$ for the $(m+1)$ items. In keeping with $\I$, the
  feasibility constraint we associate with $\I'$ is that we may sell
  item $0$, and at most one additional item from among items
  $1,\dots,m$. Note that the agent in $\I'$ is not a unit-demand
  agent.
  
  We now construct a lottery system $\Lot'$ for instance $\I'$ from $\Lot$.
  Let $\lotto=(\probi[1],\dots,\probi[m],\price)$ be a lottery in $\Lot$.
  Define $\probi[0]=\sum_{i=1}^{m}\probi$, and construct
  $\lotto'=(\probi[0],\dots,\probi[m],\price)$.  Note that
  $\lotto'$ does not necessarily satisfy the requirement that the
  $\probi$'s sum to at most one; it does, however, satisfy the
  feasible constraint indicated for $\I'$.  We may thus
  still apply the same technique as in the proof of
  Lemma~\ref{lemma:lottoVsPrice}, albeit with a worsened constant.

  Let $\Lot'$ be the system over $m+1$ services consisting of all of
  the $\lotto'$ defined as above based on $\ell\in\Lot$.  Now, for any
  setting of $t_0,\dots,t_m$, note that the the utility an agent in $\I$
  receives from a particular lottery $\lotto\in\Lot$
  is 
  \begin{equation*} 
    \sum_{i=1}^m \probi\vali - \price =
    \sum_{i=1}^{m}\probi(t_i+t_0)-\price
    = \sum_{i=0}^{m}\probi t_i-\price,
  \end{equation*} 
  precisely the utility a corresponding
  agent in $\I'$ would receive from the corresponding $\ell'\in\Lot'$.
  We thus have $\RevAlg[\Lot]=\RevAlg[\Lot']$.

  Consider applying the proof of Lemma~\ref{lemma:lottoVsPrice} to
  $\Lot'$.  Due to the less restrictive feasibility constraint
  ($\sum_{i=0}^m \probi \le 2$) we get \begin{align*} \RevAlg[\MechLP]
  &\le \RevAlg[\LMechP] + 2\RevAlg[\Vick']\\ &\le
  3\RevAlg[\Opt'], \end{align*} where the mechanisms $\LMechP$ and
  $\Vick'$ are interpreted as being in the copies setting ${\Ifc}'$
  associated with $\I'$, and $\Opt'$ is the optimal mechanism in this
  setting.  In order to prove a bound of the form desired, however, we
  need to relate a mechanism in the setting ${\Ifc}'$ to a
  deterministic one (a pricing) in $\I$.  

  The key observation is that our feasibility constraint in ${\Ifc}'$
  (carried over from $\I'$) means that $\Opt'$ may make decisions
  about allocations and prices for pseudo-agent $0$ separately from
  those for pseudo-agents $1,\dots,m$; as such, $\Opt$ effectively
  consists of two mechanisms, one serving pseudo-agent $0$ and another
  serving pseudo-agents $1,\dots,m$, both under a unit-demand
  constraint.  Now, the optimal mechanism for serving the lone
  single-parameter pseudo-agent is a pricing, and
  Theorem~\ref{thm:bmumdPricing} gives us that a mechanism serving
  pseudo-agents $1,\dots,m$ is within a factor of $2$ of a pricing on
  $m$ items; so recalling that an agent in setting $\I$ has a value of
  $\vali=t_i+t_0$ for item $i$, we can see
  that \begin{equation*} \RevAlg[\Lot]\le 3\RevAlg[\Opt] \le
  9\RevAlg[\prices], \end{equation*} where $\prices$ is the optimal
  pricing for the setting $\I$.

  In order to improve the factor from $9$ to $8$, we need to consider
  the revenue a mechanism $M$ in the setting ${\Ifc}'$ obtains from
  pseudo-agent $0$ and from pseudo-agents $1,\dots,m$; at a particular
  valuation vector ${\ts}$ denote these quantities as
  $\RevAlg[M]_{0}({\ts})$ and $\RevAlg[M]_{-0}({\ts})$,
  respectively.  Now, as previously noted, the optimal mechanism
  $\Opt$ in ${\Ifc}'$ may treat pseudo-agent $0$ independently from
  pseudo-agents $1,\dots,m$; thus, we have that any 
  mechanism $M$ in this setting must satisfy both
  $\RevAlg[M]_{0}({\ts})\le\RevAlg[\Opt]_{0}({\ts})$ and
  $\RevAlg[M]_{-0}({\ts})\le\RevAlg[\Opt]_{-0}({\ts})$.

  Since we know that $\sum_{i=1}^{m}\probi\le1$, 
  when $t_0$ is the maximum among all the $t_i$,
  Lemma~\ref{lemma:copiesBoundForRevML} implies
  \begin{equation*} 
    \RevAlg[\Lot](\ts)\le\RevAlg[\LMech]_{0}(\ts)
    +\RevAlg[\Vick]_{0}(\ts);
  \end{equation*}
  On the other hand, when one of $t_1,\dots,t_m$ takes on the maximum
  value, we end up with, for some $i$,
  \begin{equation*} 
    \RevAlg[\Lot](\ts)\le\RevAlg[\LMech]_{i}(\ts)+2\RevAlg[\Vick]_{i}(\ts),
  \end{equation*}
   Combining these two gives us 
  a pointwise guarantee of
  \begin{align*} 
    \RevAlg[\Lot](\ts)
    &\le \RevAlg[\LMech]_{0}(\ts)
    +\RevAlg[\Vick]_{0}(\ts) 
    +\RevAlg[\LMech]_{-0}(\ts)
    +2\RevAlg[\Vick]_{-0}(\ts)\\
    &\le 2\RevAlg[\Opt]_{0}(\ts)+3\RevAlg[\Opt]_{-0}(\ts).
  \end{align*}
  Therefore,
  \begin{align*}
  \RevAlg[\Lot] \le 2\RevAlg[\Opt]_{0} + 3\RevAlg[\Opt]_{-0} \le 2\RevAlg[\prices] + 6\RevAlg[\prices]
  \end{align*}
  implying the claimed bound of $8$.
\end{proof}

\section{Multi-agent setting}
\label{sec:multi}
In this section we study multi-agent versions of the BMUMD and once
again bound the gap between deterministic and randomized mechanisms
with respect to expected revenue for this setting. The starting point
for our bounds is the observation in Section~\ref{subsec:randMechAsLotteries} 
that randomized mechanisms for this problem can be interpreted as lottery-based
mechanisms. We first discuss Setting 3, namely instances with multiple
agents and multiple items and a ``matching'' feasibility
constraint. The following subsection contains a more general version
with a matroid intersection feasibility constraint (Setting 4).

\subsection{The multi-item auction setting (Setting 3)}
\label{subsec:multiItemAuctions}
We consider instances of the BMUMD where the seller has $m$ different
items, with $k_j$ copies of item $j$, and each of the $n$ unit-demand
buyers have independently distributed values for each item.  The
seller's constraint is to allocate item $j$ to no more than $k_j$
agents, and to allocate at most one item to each agent. 

We note that the set system defined by this feasibility constraint is
a matroid intersection. A set system $(E,\mathcal{F})$ where $E$ is the
ground set of elements ($E=I\times J$ in our setting) is a matroid if
it satisfies the following properties.
\begin{enumerate}
\item {\bf (heredity)} For every $A\in \mathcal{F}$, $B\subset A$ implies $B\in\mathcal{F}$.
\item {\bf (augmentation)} For every $A,B\in\mathcal{F}$ with $|A|>|B|$,
there exists an $e\in A\setminus B$ such that $B\cup\{e\}\in\mathcal{F}$.
\end{enumerate}
The sets in a matroid set system are called independent sets.

A matroid intersection set system $\mathcal{F}$ is an intersection of
two matroids: $\mathcal{F}=\mathcal{F}_1\cap\mathcal{F}_2$ where
$\mathcal{F}_1$ and $\mathcal{F}_2$ are matroids. The unit-demand
constraint and the supply constraints for each item are each instances
of a partition matroid. Thus the system $\sets$ in this setting can be
seen to be an intersection of two partition matroids. We use $\sets_1$
and $\sets_2$ to denote the two constituent matroids, and the term
matching to refer to any allocation or set in $\sets$.

We will need the following facts about matroids.
\begin{proposition}
\label{prop:matroid-bijection}
Let $B_1$ and $B_2$ be any two independent sets of equal size in some
matroid set system $\mathcal{E}$. Then there is a bijective function
$g:B_1\setminus B_2\rightarrow B_2\setminus B_1$ such that for all
$e\in B_1\setminus B_2$, $B_1\setminus\{e\}\cup\{g(e)\}$ is
independent in $\mathcal{E}$.
\end{proposition}

\begin{corollary}
\label{cor:matroid-map}
Let $B_1$ and $B_2$ be arbitrary independent sets in some matroid set
system $\mathcal{E}$. Then there exists a set $B_2'\subseteq B_2$ and
a one to one function $g:B_2'\rightarrow B_1$ such that for all $e\in
B_2'$, $B_1\setminus\{g(e)\}\cup\{e\}$ is independent in
$\mathcal{E}$, and for all $e\in B_2\setminus B_2'$, $B_1\cup\{e\}$ is
independent in $\mathcal{E}$.
\end{corollary}
\begin{proof}
  In order to apply Proposition~\ref{prop:matroid-bijection} we need
  independent sets of equal size.  So we begin by repeatedly applying
  the augmentation property to whichever of $B_1$ and $B_2$ is smaller
  in order to end up with two sets $\expand{B_1}\supset B_1$ and
  $\expand{B_2}\supset B_2$ such that
  $\abs{\expand{B_1}}=\abs{\expand{B_2}}$.  Now,
  Proposition~\ref{prop:matroid-bijection} guarantees us a bijection
  $g:\expand{B_2}\setminus\expand{B_1}\rightarrow\expand{B_1}\setminus\expand{B_2}$
  such that $\forall e\in\expand{B_2}\setminus\expand{B_1}$,
  $\expand{B_1}\setminus\{g(e)\}\cup\{e\}$ is independent.

  Set
  ${B_2}'=B_{2}\setminus\expand{B_1}\subset\expand{B_{2}}\setminus\expand{B_1}$;
  note that since $\expand{B_1}{\subset}B_1\cup B_2$, we
  have $\expand{B_1}\setminus\expand{B_2}\subset B_1$.  Thus, we may
  view $g$ as a one to one function $g:{B_2}'\rightarrow B_1$.  It
  retains the first specified property, since for any $e\in{B_2}'$,
  $B_1\setminus\{g(e)\}\cup\{e\}\subset\expand{B_1}\setminus\{g(e)\}\cup\{e\}$
  is independent.  Furthermore,
  $e{\in}B_2\setminus{B_2}'\subset\expand{B_1}$ implies
  $B_1\cup\{e\}\subset\expand{B_1}$ is independent, and so the second
  specified property holds as well.
\end{proof}



Our proof consists of three steps:
\begin{enumerate}
\item From Lemma~\ref{lem:randMechLotRedn}, we note that any
randomized mechanism for this problem can be seen as a lottery-based
mechanism.


\item We bound the revenue of any lottery-based mechanism  
for an instance $\I$ of the BMUMD by those of a collection of three
truthful deterministic mechanisms for the corresponding
single-parameter instance with copies, $\Ifc$.
\item We use the result in \cite{CHMS-10} (Theorem~\ref{thm:bmumdOPM}) to construct a truthful deterministic 
mechanism for $\I$ whose revenue is within a factor of $4/27$ of the
optimal revenue for $\Ifc$.
\end{enumerate}

\begin{lemma}\label{lem:lotSPRednMatchings}
Consider an instance $\I$ of the BMUMD in Setting 3. The revenue from
any lottery-based mechanism $M^{\lotCol}$ for $\I$ is at most five
times the expected revenue of Myerson's mechanism for the instance
$\Ifc$.
\end{lemma}
\begin{proof}
We define three truthful deterministic mechanisms $M_1,\ M_2,\ M_3$ for $\Ifc$, 
all facing the same feasibility constraint $\sets$ as the set of lottery pricings
$\lotCol$, such that 
\begin{align}\label{eqn:threeMech}
\RevAlg[M^{\lotCol}](\vals) &\leq \RevAlg[M_1](\vals) + 2\left(\RevAlg[M_2](\vals) + \RevAlg[M_3](\vals)\right)\\
\notag
&\leq 5\RevMye(\vals),
\end{align}
 The second inequality follows from the optimality of Myerson's 
mechanism for single parameter settings (Myerson's mechanism also faces 
the feasibility constraint $\sets$).


Consider the $\Ifc$ setting and fix an instantiation of values
$\vals$. Let $A_1(\vals)$ denote the set of pseudo-agents that belong
to the maximum-valued matching (we drop the argument wherever it is
obvious). Among the remaining $(I\times J)\setminus A_1$
pseudo-agents, again let $A_2$ denote the set of pseudo-agents that
belong to the maximum-valued matching i.e.
\begin{align*}
A_2(\vals) = \underset{S \text{ is a matching}}{\underset{S\subseteq [mn], S\cap A_1(\vals)=\emptyset}{\text{argmax}}}\val(S).
\end{align*}
We may assume without loss of generality that $A_1$ and $A_2$ are
defined uniquely.

Note that $A_1(\vals)$ is a unit-demand allocation
function. Therefore, Lemma~\ref{lemma:copiesBoundForRevML} implies
that
\begin{align}\label{eqn:threeMech-2}
\RevAlg[\MechL](\vals) \le \underbrace{\RevAlg[\LMech](\vals)}_{\text{Term}_1} + \underbrace{\sum_{(i,j)\notin A_1(\vals)}\probi[ij](\vals)\valij}_{\text{Term}_2}.
\end{align}


We now define the three mechanisms $M_1$, $M_2$ and $M_3$ for
$\Ifc$. Mechanism $M_1$ is $\LMech$ and so $\RevAlg[M_1]$ is exactly
Term$_1$. Mechanisms $M_2$ and $M_3$ are defined in such a way that
$2(\RevAlg[M_2] + \RevAlg[M_3])$ is at least Term$_2$. This would
prove~\eqref{eqn:threeMech}. 




Now, Corollary~\ref{cor:matroid-map} implies the existence of two one
to one partial functions with the following properties.
\begin{align*}
&g_{1}: A_2 \rightarrow A_{1} & & \text{ s.t. } \forall e \in A_2:\\ 
& & & g_1(e) \text{
is undefined and } A_1\cup\{e\} \in \sets_1, \text{ or }\\
& & & g_1(e) \text{ is defined and } A_1 \setminus \{g_{1}(e)\} \cup \{e\} \in \sets_1\\ 
&g_{2}: A_2 \rightarrow A_{1} & & \text{ s.t. } \forall e \in A_2:\\ 
& & & g_2(e) \text{ is undefined and } A_1\cup\{e\} \in \sets_2, \text{ or }\\ 
& & & g_2(e) \text{ is defined and }
A_1 \setminus \{g_{2}(e)\} \cup \{e\} \in \sets_2
\end{align*}
Note that the maximality of $A_1$ implies that every element of $A_2$
has an image under either $g_1$ or $g_2$ or both. We define the
mechanisms $M_2$ and $M_3$ by specifying their allocation rules. Given
a valuation vector $\vals$, the mechanism $M_2$ serves only those
pseudo-agents $(i,j)$ that belong to $A_1$ and for which
$\vali[ij] \geq \vali[g_1^{-1}(i,j)]/2$ (if $g_1^{-1}$ is defined at
that point). Likewise, mechanism $M_3$ serves only those pseudo-agents
$(i,j) \in A_1$ that have $\vali[ij] \geq \vali[g_2^{-1}(i,j)]/2$ (if
defined). We note that $M_2$ and $M_3$ have monotone allocation rules,
and are therefore truthful. Truthful payments can be defined
appropriately. They also satisfy the feasibility constraint $\sets$.

 
We now prove the revenue guarantee for $M_2$ and $M_3$ through the
following two claims.
\begin{claim}\label{claim:mechBoundsA2}
Twice the combined revenue of mechanisms $M_2$ and $M_3$ is no less
than the sum of values of all pseudo-agents in $A_2$, i.e.,
\begin{align*}
2\left(\RevAlg[M_2](\vals) + \RevAlg[M_3](\vals)\right) \geq \sum_{(i,j)\in A_2}\vali[ij].
\end{align*}
\end{claim}
\begin{proof}
Consider any pseudo-agent $(i,j) \in A_2$, and the pseudo-agents
$g_1(i,j)$ and $g_2(i,j) \in A_1$ if defined. Note that $A_1' =
A_1\cup {(i,j)} \setminus \{g_1(i,j), g_2(i,j)\}$ is feasible. Suppose
both $\vali[g_1(i,j)]$ and $\vali[g_2(i,j)]$ are less than
$\vali[ij]/2$; then the matching $A_1'$ is a valid matching and
$v(A_1') > v(A_1)$ which is a contradiction to the optimality of
$A_1$. Thus one of $\vali[g_1(i,j)]$ or $\vali[g_2(i,j)]$ must be at
least $\vali[ij]/2$ and we get this amount in $M_2$ or $M_3$
respectively.
\end{proof}
\begin{claim}\label{claim:A2BoundsSocVal}
The sum of values of all pseudo-agents in $A_2$ is no less than Term$_2$:
\begin{align*}
\sum_{(i,j)\in A_2(\vals)}\vali[ij] \geq \sum_{(i,j) \notin A_1(\vals)}\probi[ij](\vals)\vali[ij].
\end{align*}
\end{claim}
\begin{proof}
Consider the $n\times m$ matrix of all probabilities
$\probi[ij](\vals)$. This matrix arose from a feasible randomized
mechanism; it therefore represents a probability distribution over
matchings and can be represented as a convex combination of
matchings. In this probability matrix, replace with zeros all the
entries $(i,j) \in A_1$. The newly obtained matrix can be represented
as a convex combination of matchings all of which have a zero entry
for every $(i,j) \in A_1$. Then the claim follows by the definition of
$A_2$.
\end{proof}

Claims~\ref{claim:mechBoundsA2} and \ref{claim:A2BoundsSocVal}
together with Equations~\eqref{eqn:threeMech}
and \eqref{eqn:threeMech-2} complete the proof.
\end{proof}

\begin{theorem}
\label{thm:mp-matching}
The revenue of any randomized mechanism for an instance of the BMUMD
in Setting 3 is at most $33.75$ times the revenue of the optimal
truthful deterministic mechanism for the instance.
\end{theorem}
\begin{proof}
The proof follows from Lemmas~\ref{lem:randMechLotRedn}
and~\ref{lem:lotSPRednMatchings}, and Theorem~\ref{thm:bmumdOPM}.
\end{proof}

\subsection{The general matroid setting (Setting 4)}
We now show that Theorem~\ref{thm:mp-matching} extends to the general
matroid intersection version of the BMUMD as well. While
Lemma~\ref{lem:lotSPRednMatchings} extends to this more general
setting almost exactly, the counterpart of Theorem~\ref{thm:bmumdOPM}
for this setting is somewhat weaker. So we can only bound the gap
between the revenue of an optimal randomized incentive-compatible
mechanism and that of an optimal deterministic implementation in
undominated strategies (see Theorem~\ref{thm:bmumdSPM}) for this
setting.

As defined earlier, in Setting 4, the seller faces a feasibility
constraint specified by the set system $\sets\subseteq 2^{I\times J}$,
where $I$ is the set of agents and $J$ is the set of services, $\sets$
is the intersection of a general matroid constraint (given by
$\sets_1$) and the unit demand constraint (that we denote using
$\sets_2$); $\sets = \sets_1 \cap \sets_2$. Note that $\sets_2$ is
also a matroid.


We use the same three step approach as for the matching version to bound
the revenue of the randomized mechanism

\begin{lemma}\label{lem:lotSPRednGeneral}
Consider an instance $\I$ of the BMUMD in Setting 4.  The revenue from
any lottery-based mechanism $M^{\lotCol}$ for instance $\I$ is at most
five times the expected revenue of Myerson's mechanism for the single
parameter instance with copies $\Ifc$.
\end{lemma}
\begin{proof}
We will prove this Lemma along the lines of our proof for
Lemma~\ref{lem:lotSPRednMatchings}.  We define three truthful
deterministic mechanisms $M_1,\ M_{2},\ M_{3}$ for $\Ifc$ so that
\begin{align}\label{eqn:fiveMechGeneral}
\RevAlg[M^{\lotCol}](\vals) &\leq \RevAlg[M_1](\vals)
 + 2\left(\RevAlg[M_{2}](\vals) + \RevAlg[M_{3}](\vals) \right)\\
& \leq 5\RevMye(\vals)\nonumber.
\end{align}

As before, given an instantiation of values $\vals$, let $A_1(\vals)$
denote the set of pseudo-agents that belong to the maximum valued
feasible set.  Among the remaining pseudo-agents, let $A_2(\vals)$
denote the set of pseudo-agents that belong to the maximum valued
feasible set i.e.
\begin{align*}
A_2(\vals) = \underset{S \in \sets - A_1(\vals)}{\text{argmax}} \val(S)
\end{align*}
Lemma~\ref{lemma:copiesBoundForRevML} implies
\begin{align*}
\RevAlg[\MechL](\vals) \le \underbrace{\RevAlg[\LMech](\vals)}_{\text{Term}_1} + \underbrace{\sum_{(i,j)\notin A_1(\vals)}\probi[ij](\vals)\valij}_{\text{Term}_2}.
\end{align*}
Therefore, once again we define $M_1$ to be $\LMech$ and define $M_2$
and $M_3$ in such a way that twice their revenue combined is no less
than Term$_2$.

As before we can define partial one to one functions from $A_2$ to
$A_1$ satisfying
\begin{align*}
&g_{1}: A_2 \rightarrow A_{1} & & \text{ s.t. } \forall e \in A_2:\\ 
& & & g_1(e) \text{
is undefined and } A_1\cup\{e\} \in \sets_1, \text{ or }\\
& & & g_1(e) \text{ is defined and } A_1 \setminus \{g_{1}(e)\} \cup \{e\} \in \sets_1\\ 
&g_{2}: A_2 \rightarrow A_{1} & & \text{ s.t. } \forall e \in A_2:\\ 
& & & g_2(e) \text{ is undefined and } A_1\cup\{e\} \in \sets_2, \text{ or }\\ 
& & & g_2(e) \text{ is defined and }
A_1 \setminus \{g_{2}(e)\} \cup \{e\} \in \sets_2
\end{align*}

The mechanisms $M_2$ and $M_3$ are also defined as before: $M_2$
serves only those pseudo-agents $(i,j)$ in $A_1$ for which
$\vali[ij] \geq \vali[g_1^{-1}(i,j)]/2$ (if defined), and $M_3$ serves
only those pseudo-agents $(i,j) \in A_1$ that have
$\vali[ij] \geq \vali[g_2^{-1}(i,j)]/2$ (if defined). We note that
every element in $A_2$ gets mapped to at least one and at most two
elements under the partial functions defined above. Therefore, we can
extract a revenue of at least $1/2 \sum_{(i,j)\in A_2}\vali[ij]$ from
$M_2$ and $M_3$ together. Claim~\ref{claim:A2BoundsSocVal} now implies
the result.
\end{proof}

\begin{theorem}\label{thm:matroidThm}
The revenue of any incentive compatible randomized mechanism for an
instance $\I$ of the BMUMD in Setting 4 is at most $40$ times the
revenue of the optimal deterministic mechanism for $\I$ implemented in
undominated strategies.
\end{theorem}
\begin{proof}
The proof follows from Lemmas~\ref{lem:randMechLotRedn} and~\ref{lem:lotSPRednGeneral}, and Theorem~\ref{thm:bmumdSPM}. 
\end{proof}


\section{Discussion and open problems}
We show that in multi-parameter Bayesian optimal mechanism design the
benefit of randomness is only a small constant factor when agents are
unit-demand and their values for different items have little or no
correlation. We believe that this result should extend to instances
involving arbitrary positive correlation between values of a single
agent for items that are substitutes (the unit-demand constraint). For
example, it would be interesting to extend our result to the
multiplicative values model of Armstrong~\cite{Arm99}. Another open
problem is to extend our techniques beyond the unit-demand
setting. This may lead to a better understanding of and approximations
to optimal mechanism design in those settings, for which nothing is
known as yet.

\section*{Acknowledgements}
We thank Jason Hartline and Robert Kleinberg for many helpful
discussions.

\bibliographystyle{plain}
\bibliography{agt}

\appendix
\subsubsection*{Gap between lottery pricings and Myerson's mechanism}
We give an example where the revenue of a lottery pricing for a single
agent BMUMD instance $\I$ is $1.13$ times the revenue of Myerson's
mechanism for the instance $\Ifc$. The instance $\I$ is defined as
follows. There is a single agent with \iid valuations for two items,
distributed according to the {\it equal-revenue} distribution, bounded
at $n$.  Formally, the valuations $\vali[1]$ and $\vali[2]$ for items
1 and 2 have cdfs $\disti[1]$ and $\disti[2]$ such that
\begin{align*}
\disti[1](x) = \disti[2](x) = \begin{cases} 1-1/x & 1 \leq x < n\\
1 & x = n
\end{cases}.
\end{align*}

For the single parameter setting $\Ifc$, an upper bound on the
expected revenue of any mechanism can be obtained by removing the
feasibility constraint of allocating to a single agent at a
time. Then, the optimal revenue with the feasibility constraint is no
more than twice the optimal revenue that can be obtained by a single
agent alone. The latter, for the equal revenue distribution, is $1$
regardless of the price charged to the agent. Therefore, the optimal
revenue for $\Ifc$ is bounded above by $2$. The same bound also
applies to the revenue of any item pricing for $\I$.

Now let us consider the following lottery pricing $\Lot$ for $\I$.
\begin{align*}
\lotCol = \{(0.5, 0.5, 2.5), (1,0,2+\frac{3n}{8}), (0,1,2+\frac{3n}{8})\}
\end{align*}
The first two coordinates in every lottery denote the probabilities
with which items 1 and 2 are offered by that lottery and the third
coordinate is the price. 

Figure~\ref{fig:allocationRegion} shows the allocation function of
this lottery pricing. In particular, $R_i$ for $i\in [3]$ is the set
of valuations where lottery $i$ is bought. The probability mass of
regions $R_2$ and $R_3$ together can be computed to be $2(4/3n+O(\log
n/n^2))$. The probability mass of region $R_1$ is $0.4+0.08\ln 4 -o(1)
\approx 0.51$. Therefore, the revenue of $\lotCol$ can be computed to be
$5/2 \cdot 0.51 + 3n/8 \cdot 8/3n + o(1) = 2.275 + o(1)$.  This is a
factor of $1.13$ higher than the optimal revenue for $\Ifc$, or the
revenue of any item pricing for $\I$.

\begin{figure}[H]
        \begin{center}
        \epsfig{file=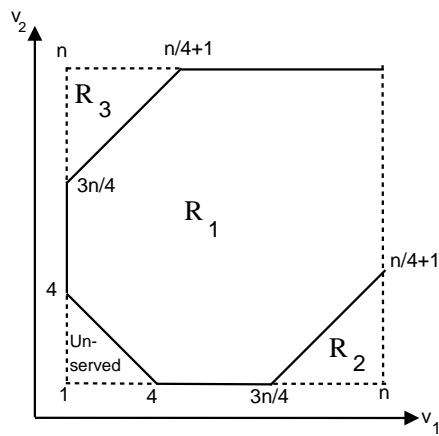,scale = 0.3}
         \end{center}
         \caption{The allocation function for the lottery pricing $\lotCol$.}         
        \label{fig:allocationRegion}
\end{figure}

\end{document}